\documentclass[american,aps,prx,reprint,floatfix,nofootinbib,superscriptaddress,notitlepage]{revtex4-2}
\usepackage{hyperref}
\hypersetup{colorlinks,linkcolor=myurlcolor,citecolor=myurlcolor,urlcolor=myurlcolor}
\usepackage{graphics,epstopdf,graphicx,amsthm,amsmath,amssymb,braket,colortbl,color,bm,framed,mathrsfs}
\usepackage[T1]{fontenc}
\usepackage[up]{subfigure}
\usepackage{tikz}
\usepackage{algorithm}
\usepackage{algorithmic}
\usepackage{enumerate}
\usepackage{tcolorbox}
\usepackage[toc,page,header]{appendix}

\definecolor{myurlcolor}{rgb}{0,0,0.9}

\theoremstyle{plain}

\newcommand*{\myproofname}{Proof}

\def\ot{\otimes}

\def\gauss{\textbf{\textit{Gauss}}}
\def\dgauss{\textbf{\textit{DGauss}}}

\DeclareMathAlphabet{\mathcal}{OMS}{cmsy}{m}{n}

\let\oldaddcontentsline\addcontentsline
\newcommand{\stoptocentries}{\renewcommand{\addcontentsline}[3]{}}
\newcommand{\starttocentries}{\let\addcontentsline\oldaddcontentsline}
\stoptocentries

\newtheorem{theorem}{Theorem}
\newtheorem{definition}{Definition}

\newtheorem{lemma}[theorem]{Lemma}
\newtheorem{proposition}[theorem]{Proposition}

\newtheorem{corollary}{Corollary}

\numberwithin{equation}{section}
\numberwithin{corollary}{section}
\numberwithin{definition}{section}
\numberwithin{theorem}{section}
\numberwithin{remark}{section}
\numberwithin{example}{section}

\newcommand{\ra}{\rangle}
\newcommand{\la}{\langle}
\newcommand{\df}{\dfrac}

\newcommand{\pd}[1]{\partial_{#1}} 
\newcommand{\R}{\mathbb R}

\newcommand{\mrm}{\mathrm}

\newcommand{\mfk}{\mathfrak}

\newcommand{\tr}{\mathrm{Tr}}
\newcommand{\mca}{\mathcal}
\newcommand{\Hi}{\mca H}
\newcommand{\G}{\mca G}

\newcommand{\C}{\mathbb C}

\newcommand{\Pf}{\mathrm{Pf}}
\newcommand{\Cl}{\mca C}

\newcommand{\leqalign}[2]{
\begin{equation}\begin{aligned}\label{#1}
#2
\end{aligned}\end{equation}}
\newcommand{\malign}[1]{
\[\begin{aligned}
#1
\end{aligned}\] 
}

\makeatother
\begin{document}

\author{Xingjian Lyu}
\email{nicholaslyu@college.harvard.edu}
\affiliation{ Harvard University, Cambridge, Massachusetts 02138, USA}

\author{Kaifeng Bu}
\email{bu.115@osu.edu}
\affiliation{Ohio State University, Columbus, Ohio 43210, USA}
\affiliation{ Harvard University, Cambridge, Massachusetts 02138, USA}

\title{Displaced Fermionic Gaussian States and their Classical Simulation}

\begin{abstract} 
    This work explores displaced fermionic Gaussian operators 
    with nonzero linear terms. We first demonstrate equivalence between several characterizations 
    of displaced Gaussian states. We also provide an efficient classical simulation protocol 
    for displaced Gaussian circuits and demonstrate 
    their computational equivalence to circuits 
    composed of nearest-neighbor matchgates augmented by 
    single-qubit gates on the initial line. Finally, we 
    construct a novel Gaussianity-preserving unitary embedding 
    that maps $n$-qubit displaced Gaussian states to $(n+1)$-qubit even Gaussian states. 
    This embedding facilitates the generalization of existing Gaussian 
    testing protocols to displaced Gaussian states and unitaries. 
    Our results provide new tools to analyze fermionic systems 
    beyond the constraints of parity super-selection, extending 
    the theoretical understanding and practical simulation of 
    fermionic quantum computation. 
\end{abstract}

\maketitle
\section{Introduction}
\stoptocentries

Fermionic quantum computation is a promising model 
for implementing universal quantum computation. 
In this model, fermionic Gaussian states---corresponding 
to free fermions---play a central role. 
Fermionic Gaussian states arise in the study of 
a variety of systems, such as the 1-dimensional Ising model, 
or fermionic linear optics with beam splitters and 
phase shifters acting on non-interacting 
electrons \cite{divincenzo2005fermionic}. They are also 
foundational in computational chemistry, notably within the 
Hartree-Fock method for modeling molecular 
orbitals \cite{surace2022fermionic}. 

In addition, fermionic Gaussian circuits are closely related 
to matchgate circuits, one of the few known models 
enabling efficient classical simulation 
of quantum computations \cite{valiant2001quantum}. 
The relation between fermionic Gaussian circuits and matchgate 
circuits has been explored in~\cite{divincenzo2004fermionic,jozsa2008matchgates}, 
and subsequent works explored the classical simulation of these circuits 
using Grassmann integral methods, topics which were extensively studied in 
mathematics and field theory \cite{bravyi2004lagrangian, bravyi2019simulation, 
soper1978construction}.  
Beyond fermionic Gaussian circuits, the Gaussian framework extends 
to other classically simulable circuits, including 
continuous variable Gaussian circuits~\cite{weedbrook2012gaussian} 
and stabilizer circuits~\cite{beny2022gaussian}, also 
known as discrete quantum Gaussian circuits~\cite{BGJ23a,BGJ23b,bu2024extremality,BJ24a,BGJ24a,Bu19,bu2022classical}.

In fermionic systems, Gaussian unitaries are generated 
by quadratic Hamiltonians without linear terms, reflecting the 
dynamics of non-interacting fermions. A parity super-selection 
rule in physical fermionic systems constrains states to 
definite parity and excludes Hamiltonians with 
odd-degree terms \cite{cahill1999density, szalay2021fermionic}; 
this constraint is also linked to the symmetry of fermion number 
conservation, whose observable manifests as the parity operator. 
Consequently, the study of fermionic Gaussian 
computation has predominantly focused on even Gaussian 
operators, which also conveniently exhibit a more manageable mathematical 
structure due to commutativity properties. However, 
in representations of fermionic algebra on non-fermionic 
platforms, such as 1-dimensional spin chains or qubit systems 
via the Jordan-Wigner transformation, the parity super-selection 
rule is more a mathematical convenience than physical necessity. 
It limits the scope of Gaussian analysis by excluding 
certain operators and states; this impacts 
applications like fermionic convolution and variational 
characterizations of Gaussianity \cite{lyu2024fermionic}.

This work addresses the limitations of conventional fermionic 
Gaussian theory imposed by the parity super-selection rule. 
To broaden the applicability of Gaussian theory in fermionic 
systems, we focus on displaced Gaussian states and 
unitaries—those with nontrivial linear terms in addition to 
quadratic ones. A high-level mathematical reduction, connecting 
the Lie algebras of displaced and even Gaussian operators, 
was proposed in \cite{knill2001fermionic}, and a previous work has 
explored the channel capacities of 
displaced Gaussian states and channels \cite{bravyi2005classical}. 
However, a framework relating displaced Gaussian 
computation to the more extensively studied even 
Gaussian (or matchgate) computation remains undeveloped. 
This gap is significant given that even Gaussian 
computation has been deeply investigated in areas 
such as magic states, classical simulability, and resource 
theory \cite{hebenstreit2019all, hebenstreit2020computational, 
dias2024classical, cudby2023gaussian, reardon2023improved}. 
Throughout this work, we use “Gaussian” to refer specifically 
to zero-mean cases, reserving “displaced Gaussian” for 
the general non-zero mean scenarios.

This work extends the scope of fermionic Gaussian theory by 
examining displaced Gaussian operators and situating them within 
the broader landscape of even Gaussian operators. 
An initial challenge which motivated this study was the 
ambiguity in defining displaced Gaussian states 
across differing interpretations of Gaussianity. 
Definitions of even Gaussian states vary by context: 
computation-oriented works define them as the orbit of 
computational basis states under Gaussian 
unitaries \cite{cudby2023gaussian, dias2024classical}, 
while physics-motivated works identify them as the thermal states of 
quadratic Hamiltonians \cite{surace2022fermionic, bianchi2021page}. 
Yet another approach rooted in phase-space formulation 
identifies Gaussian states with the quadratic form of 
their Grassmann representations \cite{bravyi2004lagrangian, 
bravyi2005classical}. Extending these definitions to displaced 
Gaussian states raises consistency challenges. 
Additionally, from the perspective of classical simulation, 
despite the prior work on the Lie 
algebra embeddings suggesting the available efficient 
classical simulability of displaced Gaussian computation \cite{knill2001fermionic},
an explicit reduction from displaced to even Gaussian computation has 
yet to be fully articulated.

The main contributions of this work are as follows: 
first, we unify the displaced extensions of the 
aforementioned definitions of Gaussian states. 
Second, we introduce an efficient classical simulation algorithm 
for displaced Gaussian circuits and prove that displaced 
Gaussian computation is generated by nearest-neighbor 
matchgates augmented with single-qubit gates on the initial line. 
Finally, we construct a novel unitary embedding that maps 
displaced Gaussian states into even Gaussian states, 
providing a bridge between the study of even and displaced 
fermionic Gaussian operators. 
Leveraging this embedding, we extend recent 
findings \cite{lyu2024fermionic} to develop 
operational tests for identifying classically 
simulable displaced Gaussian states and unitaries.

Section~\ref{sec:prelim} introduces displaced Gaussian operators 
and provides a review of the Jordan-Wigner transformation, 
matchgates, and existing results for even fermionic Gaussian operators. 
Section~\ref{sec:dgTheory} examines the main properties of displaced 
Gaussian unitaries and states. Section~\ref{sec:dgSimulation}
presents a simulation protocol for displaced Gaussian circuits, 
highlighting the practical implications of our findings. 
Finally, Section~\ref{sec:dgTesting} details the construction 
of the even Gaussian embedding and illustrates its 
application in constructing tests for displaced Gaussian components. 

\section{Preliminaries}
\label{sec:prelim}
A system of $n$ fermionic modes is associated 
with $2n$ Hermitian, traceless and 
Majorana operators $\{\gamma_j\}_{j=1}^{2n}$ 
which generate a Clifford algebra $\Cl_{2n}$ according to 
the anticommutation relations 
\begin{eqnarray}
    \{\gamma_j, \gamma_k\} = 2\delta_{jk}I. 
\end{eqnarray}
They are related to the real and imaginary parts of the 
creation and annihilation operators by 
\begin{eqnarray}
    \gamma_{2j-1} = a_j + a_j^\dag, \quad 
    \gamma_{2j} = i(a_j - a_j^\dag). 
\end{eqnarray}
Using the Jordan-Wigner transform~\cite{jordan1993paulische}, 
the Majorana operators 
can also be identified with products of 
Pauli operators on the operator space $\mca H_n$ 
of $n$ qubits: 
\leqalign{eq:jwTransform}{ 
    \gamma_{2j-1} 
    &= Z^{\ot(j-1)} \ot X \ot I^{\ot (n-j)}, \\ 
    \gamma_{2j} 
    &= Z^{\ot (j-1)} \ot Y \ot I^{\ot (n-j)}. 
}
Since permuting the tensor products 
does not change the anticommutation relations, 
the representation remains valid 
if we permute the qubits $1, \dots, n$. 
We identify the \textit{initial line} of the circuit 
with the first subspace in equation~\ref{eq:jwTransform}, 
and two qubits are \textit{nearest neighbors} if 
they are adjacent in the tensor product. 

Given an ordered tuple (multi-index) $J\subset [2n]$, 
we denote the size of the tuple $J$ by $|J|$
and define $\gamma_J$ to be the ordered product 
$\gamma_J = \prod_{j\in J} \gamma_j$ indexed by $J$. 
The products $\{\gamma_J\}_{J\subset [2n]}$ form an 
orthonormal basis, so every operator 
$A\in \Cl_{2n}\cong \Hi_n$ has a Majorana expansion~\cite{jaffe2015reflection}: 
\leqalign{def:majoranaExpansionMain}{ 
    A = \df 1 {2^n} \sum_{J\subset [2n]} A_J \gamma_J, \quad 
    A_J = \tr(\gamma_J^\dag A)\in \C. 
}
The coefficients $\{A_J\}$ are also called the ``moments'' 
of $A$. It is also useful to introduce $2n$ self-adjoint Grassmann 
generators $\{\eta_j\}_{j=1}^{2n}$ which generate a 
Grassmann algebra by the anticommutation relation
$\{\eta_j, \eta_k\} = 0$. The Fourier transform 
$\Xi_A(\eta)\in \G_{2n}$ is computed by multiplying $2^n$ and 
substituting the generators in
 equation~\ref{def:majoranaExpansionMain}, 
with $\eta_J$ defined analogously: 
\begin{eqnarray}
    \Xi_A(\eta) = \sum_{J\subset [2n]} A_J \eta_J. 
\end{eqnarray}
We identify the mean and the covariance of an operator $A\in \Cl_{2n}$ 
with the first and second-order moments $\{A_j\}_{j=1}^{2n} = \{\mu(A)_j\}$, 
$\{A_{jk}\}_{j, k=1}^{2n} = \{\Sigma(A)_{jk}\}$. 
These define the antisymmetric \textit{extended covariance matrix} of $A$ 
\leqalign{def:extCovMatrixMain}{
    \tilde \Sigma(A) = \begin{bmatrix}
        \Sigma(A) & i\mu(A) \\ -i\mu(A)^T & 0 
    \end{bmatrix} \in \mfk{so}(2n+1, \C). 
}
Note that $\tilde \Sigma(\rho), \Sigma(\rho)$ are purely 
imaginary for a Hermitian state $\rho$. They are 
normalized so that $\Sigma(\rho \otimes I) = \Sigma(\rho) \oplus 0_{2\times 2}$. 
It is also convenient to define the raw extended covariance matrix 
\leqalign{def:rawExtCovMain}{
    \bar \Sigma(A) = \df 1 {2^n} \tilde \Sigma(A). 
}
Given a multi-index $J\subset [2n+1]$, 
we denote the $|J|\times |J|$ antisymmetric matrix 
restriction of $\tilde \Sigma(J)$ onto the subspaces indexed by $J$
according to  
\leqalign{def:subspaceRestrictionMain}{
    \left[
        \tilde \Sigma(\rho)_{|J}
    \right]_{ab} = \tilde \Sigma(\rho)_{J_aJ_b}. 
}

\begin{definition}[displaced Gaussian unitary]
    $U\in \Cl_{2n}$ is a displaced Gaussian unitary if 
    \leqalign{def:DGMain}{
    U = \exp \left(\frac 1 2 \gamma^T h \gamma + id^T\gamma\right)
    }
    where $h\in\mfk{so}(2n, \R)$ is real, antisymmetric 
    and $d\in \R^{2n}$. Here $U$ is a (even) Gaussian unitary if $d=0$. 
\end{definition}
We denote the group of even or displaced Gaussian unitaries 
on $n$ qubits by $G(n)$ or $DG(n) \subset \Cl_{2n}$, respectively. 
It is known that $U\in G(n)$ conjugates the Majorana operators by 
a rotation \cite[theorem 3]{jozsa2008matchgates}: 
\leqalign{thm:evenGeneratorRotMain}{
    U\gamma_j U^\dag = \sum_{k=1}^{2n} R_{kj}\gamma_k, \quad 
    R=e^{2h}. 
}
Here $\log U=iH$ is understood to 
be quadratic in the Majorana operators. 
Among the Gaussian unitaries, a special subset
which only act nontrivially on two nearest-neighbor lines 
are called \textit{nearest-neighbor (n.n.) matchgates}. 
Every Gaussian unitary on $n$ lines has a local circuit 
decomposition into $O(n^3)$ n.n matchgates~\cite[theorem 5]{jozsa2008matchgates}, 
each effecting rotations which act nontrivially on subspaces 
$4m-3, \dots, 4m$ or $4m-1, \dots, 4m+2$. 

The following definition can be found in~\cite{bravyi2005classical}. 
Observe that $iM = \Sigma(\rho)$ and $d=\mu(\rho)$. 
\begin{definition}[displaced Gaussian state]
    A $n$-qubit state 
    $\rho \in \Cl_{2n}$ is a displaced Gaussian state 
    if its 
    Fourier transform admits a Gaussian expression: 
    \leqalign{eq:defDgStateMain}{
        \Xi_\rho(\eta) = \exp \left(
            \df i 2 \eta^T M \eta + d^T \eta 
        \right)
    }
    where $M\in \mfk{so}(2n, \R)$ and $d\in \R^{2n}$. 
    It is a (even) Gaussian state if its mean $d=0$. 
\end{definition}
We denote the set of even and displaced 
Gaussian states by $\gauss(n)$ and 
$\dgauss(n) \subset \Cl_{2n}$, respectively. 
A displaced Gaussian state is 
completely characterized by its extended covariance matrix 
(equation~\ref{def:extCovMatrixMain}). 
A special class of \textit{diagonalized Gaussian states} 
consists of separable products of computational basis states: 
for some set of $\lambda_j\in [-1, 1]$, 
\leqalign{def:diagGaussianMain}{
    \rho_D
    &= \bigotimes_{j=1}^n \df{1 + \lambda_j Z}{2}, 
    \quad \Sigma(\rho_D) = \bigoplus_{j=1}^n \begin{pmatrix}
        0 & \lambda_j \\ -\lambda_j & 0 
    \end{pmatrix}. 
}
Given $A\in \mfk{so}(2n, \C)$, 
there exists $\rho \in \gauss(n)$ such that $\Sigma(\rho)=A$ 
if and only if $-A^TA\leq I$~\cite{bravyi2004lagrangian}. 

\section{Displaced Gaussian theory}
\label{sec:dgTheory}

In this part we describe the main results about 
displaced Gaussian states and unitaries which generalize 
the even Gaussian counterparts. In particular, we see 
that every $\rho \in \dgauss(n)$ is identified with 
the antisymmetric extended covariance $\tilde \Sigma$, 
upon which $U\in DG(n)$ conjugates by a rotation. 
These results rigorously substantiate the intuition 
that the nonzero mean of displaced Gaussian 
operators can be effectively treated 
as the covariance on one more mode. 

One advantage of defining Gaussian states 
according to its Fourier property is that 
expanding the exponential in equation~\ref{eq:defDgStateMain} 
yields the following extended Wick's formula. 
Recall equation~\ref{def:subspaceRestrictionMain}
for the following result: 
given $J\subset[2n]$, let $\tilde J=J\cup \{2n+1\}$ if $|J|$ 
is odd else $J$; 
also define the scalar factor 
$\alpha_{|J|} = (-i)^{|J|\,\operatorname*{mod}\, 2}$. 
The following result is the displaced 
generalization of Wick's formula, allowing one to deterine 
all higher moments of a displaced Gaussian 
state using its mean and covariance. 

\begin{proposition}
    \label{prp:extWickMain}
    Every $\rho\in \dgauss(n)$ satisfies 
    \leqalign{eq:extWickMain}{
        \rho = \df 1 {2^n}\sum_{J\subset[2n]} \alpha_{|J|} 
        \Pf \left[
            \tilde \Sigma(\rho)_{|\tilde J}
        \right]\gamma_J. 
    }
\end{proposition}

A combinatorial proof is provided 
in Appendix~\ref{app:extWick}. 
Using the Lie algebra embedding 
identified by Knill~\cite{knill2001fermionic}, we 
derive the following displaced generalization of 
equation~\ref{thm:evenGeneratorRotMain}, in which 
$\log U=iH$ is proportional to the quadratic Hamiltonian, 
and $\bar \Sigma$ denotes the raw extended covariance matrix 
in equation~\ref{def:rawExtCovMain}. 
\begin{theorem}
    \label{thm:dgUnitaryActionGaussian}
    Given $U\in DG(n), \rho \in \dgauss(n)$ 
    the extended covariance of $U\rho U^\dag\in \dgauss(n)$ is 
    \leqalign{}{
        \tilde \Sigma(U\rho U^\dag) = R\, \tilde \Sigma(\rho) R^T, \quad 
        R = e^{2\bar \Sigma(\log U)}. 
    }
\end{theorem}
The detailed proof can be found in Appendix~\ref{app:unitaryAction}. 
The main insight here is that under $DG(n)$, 
the mean $\mu_j \gamma_j$ of a displaced Gaussian state 
transforms as the covariance $i\mu_j\gamma_j\gamma_{2n+1}$ on one more mode. 

Even Gaussian states have been defined 
as those admitting a Gaussian Fourier 
expression~\cite{bravyi2004lagrangian,bravyi2005classical}, 
the orbit of computational basis states under
$G(n)$~\cite{cudby2023gaussian, dias2024classical}, 
or thermal states of purely quadratic 
Hamiltonians~\cite{surace2022fermionic, bianchi2021page}. 
Our definition of $\dgauss$ extends the first definition, 
yet it is equally plausible to adopt the 
extensions of the circuit or thermal state 
definitions. Using theorem~\ref{thm:dgUnitaryActionGaussian}, 
we show that all three definitions are in fact equivalent. 
This result bridges the physical, computational, and mathematical 
properties of displaced Gaussian states. 

\begin{theorem}
    \label{thm:unifiedCharacterization}
    $\rho \in \Cl_{2n}$ is a displaced Gaussian state 
    (equation~\ref{eq:defDgStateMain})
    iff it satisfies any of the following: 
    \begin{enumerate}
        \item \textbf{Thermal state}: 
        there exists a quadratic Hamiltonian $H\in \Cl_{2n}$ such that 
        $\rho$ is its thermal state: 
        \leqalign{}{
            \rho = \df{e^{-H}}{\tr(e^{-H})}, \quad 
            H = \df i 2 \gamma^T h \gamma + d^T\gamma. 
        }
        Here $h\in \mfk{so}(2n, \R)$ is real antisymmetric 
        and $d\in \R^{2n}$. 
        A pure state $\rho$ is a displaced Gaussian state iff it 
        is the ground state of some $H$ of the form above. 
        \item \textbf{Circuit output}: 
        $\rho$ results from applying a displaced 
        Gaussian unitary $U_G\in DG(n)$ 
        to a diagonalized Gaussian state $\rho_D$ 
        (equation \ref{def:diagGaussianMain}): 
        \leqalign{}{
            \rho = U_G \rho_D U_G^\dag. 
        }
    \end{enumerate}
\end{theorem}
The detailed proof in Appendix~\ref{app:dgStateCharacterization} 
first establishes equivalence for diagonalized 
Gaussian states, then uses theorem~\ref{thm:dgUnitaryActionGaussian} 
to extend to the general case. 
One immediate corollary is the necessary and sufficient conditions 
for an antisymmetric matrix to be the extended covariance of some 
$\rho \in \dgauss(n)$. This extends the covariance 
condition established for even Gaussian 
states in~\cite{bravyi2004lagrangian}. 
\begin{corollary}
    $A\in \mfk{so}(2n+1, \C)=\tilde \Sigma(\rho_G)$ 
    for some $\rho_G\in \dgauss(n)$ iff $A$ has rank $2n$ and 
    there exists $R\in SO(2n+1)$ such that 
    $RA^2R^T\leq I_{2n+1}$, the associated $\rho_G$ is pure 
    iff the nonzero eigenvalues of $A^2$ are all $1$. 
\end{corollary}

\section{Classical simulation}
\label{sec:dgSimulation}
An important motivation for studying displaced Gaussian unitaries 
is that they yield a larger class of efficiently classically 
simulable quantum circuits. 
The circuit characterization of $\dgauss(n)$ in 
theorem~\ref{thm:unifiedCharacterization} also relates 
the study of displaced Gaussian unitaries to the study of displaced 
Gaussian states. 
In this section, we show that displaced Gaussian circuits are 
equivalent to n.n. matchgates augmented with 
single-qubit gates on the initial line. 
We also describe the efficient classical 
simulation of displaced Gaussian circuits. 
Detailed derivations of the following results can be found in 
Appendix~\ref{app:classicalSimulation}. 

\begin{theorem}
    \label{thm:dispDecomposition}
    Every displaced Gaussian unitary $U\in DG(n)$ is the 
    product of $O(n^3)$ matchgates or single-qubit gates
    on the initial line of the circuit. 
\end{theorem}
Recalling theorem~\ref{thm:dgUnitaryActionGaussian}, a quadratic 
term $\gamma_j\gamma_k$ for $1\leq j <k \leq 2n$ 
in the Hamiltonian generates
rotation between the $j$ and $k$-th subspaces of the extended covariance matrix. 
On the initial line of the Jordan-Wigner transform, 
the $X$-rotation unitary $R_X(\theta)_1$ is generated by $\gamma_1=X_1$, 
which generates rotation between the 
first and $(2n+1)$-th subspaces of the extended covariance matrix. 
The theorem follows by noting that every rotation in $\R^{2n+1}$ can be decomposed 
into $O(n^3)$ rotations generated by n.n. matchgates or $R_X(\theta)_1$. 

In light of the circuit characterization in theorem~\ref{thm:unifiedCharacterization}, 
displaced Gaussian states can be understood as the orbit of 
separable computational basis mixtures under $DG(n)$. 
By slightly adapting the technique from \cite[theorem 3]{brod2016efficient}, 
we obtain the following result:

\begin{proposition} 
    \label{prp:productsAreDg}
    Every $n$-qubit product state is a displaced Gaussian state. 
\end{proposition}

To complete efficient simulation, it remains to simulate measurements. 
This is facilitated by the Grassman Gaussian integrals which have been studied 
in~\cite{bravyi2004lagrangian, bravyi2019simulation, soper1978construction}. 
We can efficiently simulate computational basis measurements 
on any subset of qubit lines: 
given a subset $K \subset [n]$ of lines to measure, where $|K|=k$, 
and a bit string $x \in {0, 1}^k$, 
the associated measurement operator is defined as: 
\leqalign{}{
    O(K, x) = \df 1 {2^k} \prod_{j=1}^k I + (-1)^{x_j} Z_{K_j}. 
}
This operator acts as $|x\rangle \langle x|$ on the qubit lines indexed by $K$, 
while it applies the identity $I$ on all other lines.
\begin{lemma}
    \label{lem:dgMeasurementMain}
    Given $\rho \in \dgauss(n)$, the expectation value of 
    measuring $O(K, x)$ is 
    \leqalign{}{
        \tr[O(K, x)\rho] = \df 1 {2^k} \sqrt{\det\left[
            1 + \Sigma(\rho) M
        \right]}
    }
    where $M=\Sigma[2^{n-k}O(K, x)]$ is the covariance matrix of the 
    even Gaussian state proportional to $O$. 
\end{lemma}
Given a displaced Gaussian state as input, 
the effect of displaced Gaussian unitaries can be determined via 
theorem~\ref{thm:dgUnitaryActionGaussian}, 
and the measurement outcomes in the computational 
basis can be computed using lemma~\ref{lem:dgMeasurementMain}.

\begin{theorem} 
    A circuit consisting of displaced Gaussian unitaries, 
    a displaced Gaussian state input, 
    and computational basis measurements 
    on any subset of qubit lines is efficiently classically simulable. 
\end{theorem}

Since both n.n. matchgates and 
single-qubit gates acting on the first line are displaced 
Gaussian unitaries, and product states qualify as displaced 
Gaussian states, we can conclude the following:

\begin{corollary} 
    \label{cor:dgSimulationExt}
    A circuit comprising n.n. matchgates 
    and single-qubit gates on the first line, 
    with a product state input and computational basis measurements 
    on any subset of qubit lines, 
    is efficiently classically simulable. 
\end{corollary}

Our work extends previous results on the classical simulation of fermionic even Gaussian 
circuits \cite{jozsa2008matchgates,terhal2002classical,brod2016efficient} by extending 
the unitaries to include displaced linear terms. This extension broadens the set of circuits 
which can be efficiently classically simulated, and we further showed that this generalization 
is equivalent to allowing arbitrary single-qubit gates on the first line.

\section{Displaced Gaussian testing}
\label{sec:dgTesting}
Previous works have proposed a channel embedding that maps 
$\rho \in \dgauss(n)$ to $\rho' \in \gauss(n+1)$~\cite{bravyi2005classical,knill2001fermionic}. 
Such embeddings are fundamental to extending even Gaussian properties 
to displaced Gaussian states. 
However, the previously studied embedding channel is not purity-preserving 
as $S(\rho') = S(\rho) + 1$, where $S(\rho)$ denotes the von Neumann entropy in bits. 
This lack of purity invariance poses challenges for extending 
computational protocols that rely on purity, 
such as the fermionic Gaussian test in~\cite{lyu2024fermionic}. 

In this section, we provide the novel construction of a unitary 
embedding, $\mca E:\dgauss(n)\to\gauss(n+1)$. 
We next outline operational protocols for 
distinguishing the efficiently classically simulable displaced 
Gaussian states and unitaries. 

\begin{definition}[even embedding channel]
    The even embedding channel 
    $\mca E:\Cl_{2n}\to \Cl_{2n+2}$ is defined by 
    \leqalign{def:embeddingUnitary}{
        \mca E(\rho) = V(\rho \ot |+\ra \la +|)V^\dag, \quad 
        V = \exp \left(-i \df \pi 4 \gamma_{2n+2}\right)
    }
\end{definition}
In Appendix~\ref{app:unitaryEmbedding}, 
we provide an implementation of 
$V$ using elementary gates. 
Given the $(2n+1)$-sized extended covariance matrix $\tilde \Sigma(\rho)$ of 
an $n$-qubit state $\rho$, by a standard matrix theorem~\cite{zumino1962normal} 
there exists a rotation $R$ such that for each $\lambda_j\in [-1, 1]$, 
\leqalign{}{
    R\, \tilde \Sigma(\rho) R^T = \left[\bigoplus_{j=1}^n \begin{pmatrix}
        0 & -i\lambda_j \\ i\lambda_j & 0 
    \end{pmatrix}\right] \oplus (0). 
}
Block-decompose $R$ into $R_0\in \R^{2n\times 2n}$, $c\in \R$, and 
$s, r\in \R^{2n}$ so that 
$R = \begin{bmatrix}
        R_0 & s \\ r^T & c 
    \end{bmatrix}$. We obtain the following result: 
\begin{theorem}
    \label{thm:unitaryEmbeddingMain}
    An $n$-qubit state $\rho$ is displaced Gaussian 
    iff the $(n+1)$-qubit state $\mca E(\rho)$ is even Gaussian, 
    in which case the covariance matrix of the embedded state is 
    \leqalign{}{
        \Sigma[\mca E(\rho)] = \begin{bmatrix}
            \Sigma(\rho) & -ir & i\mu(\rho) \\ 
            ir^T & 0 & ic \\ -i\mu(\rho)^T & -ic & 0 
        \end{bmatrix} \in \mfk{so}(2n+2, \C). 
    }
\end{theorem}

The even embedding for states has a counterpart for unitaries: 
given an $n$-qubit unitary $U\in \Cl_{2n}$, define 
\begin{eqnarray}
    \tilde U = V (U\otimes I)V^\dag \in \Cl_{2n+2}. 
\end{eqnarray}
The two embeddings are compatible via the equation 
\[ 
    \mca E(U\rho U^\dag) = \tilde U \mca E(\rho) \tilde U^\dag. 
\] 
\begin{lemma}
    \label{lem:unitaryEmbeddingMain}
    An $n$-qubit unitary $U$ is displaced Gaussian iff 
    $\tilde U$ is an $(n+1)$-qubit even Gaussian unitary. 
    Given $U$ as in equation \ref{def:DGMain}, $\tilde U$ 
    is given by 
    \leqalign{}{
        \tilde U = \exp \left(\df 1 2 \gamma^T \tilde h \gamma\right), \quad 
        \tilde h = \begin{bmatrix}
            h & 0_{2n\times 1} & -d \\ 0_{1\times 2n} & 0 & 0 \\ 
            d^T & 0 & 0  
        \end{bmatrix}. 
    }
\end{lemma}
One important tool in the study of fermionic Gaussian 
states is fermionic convolution~\cite{lyu2024fermionic}. 
The $n$-qubit convolution unitary acting $2n$ qubits is defined by 
\leqalign{def:convUnitary}{
    W = \exp \left(
        \df \pi 4 \sum_{j=1}^{2n} \gamma_j \gamma_{2n+j}
    \right) \in \Cl_{4n}. 
}
The fermionic convolution of two \textit{even} states $\rho, \sigma$ is 
\leqalign{def:fermionicConvolutionMain}{
    \rho \boxtimes \sigma = \tr_2[W(\rho \ot \sigma) W^\dag]. 
}
Here $\tr_2$ denotes partial trace over the registers of $\sigma$. 
Fermionic convolution provides valuable tests for 
Gaussian states and unitaries. 
In the results below, denote by \(\rho_{\mathcal{E}}\) 
the fermionic maximally entangled state as defined in 
\cite[Definition 6]{bravyi2004lagrangian}. The following protocols hold 
if \(|\psi\rangle\) and \(U\) (with \(U \in \mathcal{C}_{2n}\)) 
are both \textit{even} operators: 

\begin{enumerate}
    \item \textbf{Gaussian State Test}: An even state \(|\psi\rangle\) is Gaussian iff 
    the overlap \(\langle \psi | (|\psi\rangle \boxtimes |\psi\rangle) = 1\). This overlap 
    is experimentally accessible through the swap test.
    \[
        |\psi\ra \la \psi| \in \gauss(n) \iff \langle \psi | (|\psi\rangle \boxtimes |\psi\rangle) = 1
    \]
    \item \textbf{Gaussian Unitary Test}: An even unitary \(U\) is Gaussian 
    iff its fermionic Choi state is Gaussian: 
    \[ 
        U\in G(n)\iff (U \otimes I_n)\rho_{\mathcal{E}}(U^\dagger \otimes I_n)\in \gauss(2n). 
    \]
\end{enumerate}
Using Gaussianity-preserving properties 
established in theorem \ref{thm:unitaryEmbeddingMain} and 
lemma \ref{lem:unitaryEmbeddingMain}, 
we obtain tests for 
the classically simulable displaced Gaussian components, 
generalizing the results in \cite{lyu2024fermionic}. 
\begin{theorem}
    For any $n$-qubit $\rho, U\in \Cl_{2n}$: 
    \begin{itemize}
        \item $\rho\in \dgauss(n)$ iff 
        $\mca E(\rho)$ passes the test for even Gaussian states. 
        \item $U\in DG(n)$ iff 
        $V(U\otimes I)V^\dag$ passes the test for even Gaussian unitaries. 
    \end{itemize}
\end{theorem}

\section{Conclusion and discussion}
This work presented tools and results 
on displaced Gaussian operators and their computational significance, 
extending our understanding of fermionic Gaussian computation 
beyond the fermionic parity constraint. 

We characterized displaced Gaussian circuits as 
n.n. matchgates augmented with single-qubit 
gates on the initial line, unified physically and 
computationally motivated characterizations of 
displaced Gaussian states, outlined an efficient classical 
simulation protocol for displaced Gaussian circuits, and 
provided the novel construction of a 
unitary embedding of displaced Gaussian states.
These results generalized the well-studied 
properties of the even Gaussian operators beyond the fermionic 
parity constraint and highlights the significance of 
displaced Gaussian operators. 

Future works may consider extending 
classical simulability to more flexible circuit topologies,
such as allowing single-line gates on non-initial lines. 
Additionally, the complexity of relating displaced and even 
fermionic Gaussian computation motivates further theoretical 
study of particle-number symmetry's role in classical simulation, 
with potential applications to simulation protocols 
inspired by other physical particle types.

\section{Acknowledgments} 
We thank Arthur Jaffe and Liyuan 
Chen for valuable discussions. 
This work is supported in part by the ARO Grant W911NF-19-1-0302 
and the ARO MURI Grant W911NF-20-1-0082. 

\bibliography{reference}{}

\clearpage
\newpage
\onecolumngrid
\starttocentries
\appendix
\tableofcontents
\addtocontents{toc}{\protect\setcounter{tocdepth}{2}}
\addcontentsline{toc}{section}{Appendix}

\bigskip 

\section{Lie theory of Fermionic Gaussian operators}
The Clifford algebra $\Cl_{2n}$ forms a Lie algebra under the 
commutator bracket $[\gamma_J, \gamma_K]$. In this section, 
we examine the subalgebras related to Gaussian and displaced 
Gaussian unitaries, originally defined in~\cite{knill2001fermionic}. 
These constructions are foundational, as they support many of the 
later results and proofs presented in this work. 
\begin{definition}[Gaussian Lie algebras and groups]
    Let $\mfk{dg}(n)$ denote the 
    Lie algebra of quadratic polynomials 
    \begin{eqnarray}
        \mfk{dg}(n) = \{
            a \gamma_j + b \gamma_k \gamma_l\in \Cl_{2n} | 
            \{a, b\} \subset \C, \{j, k, l\}\subset [2n]
        \}, \quad \dim \mfk{dg}(n) = 2n^2 + n. 
    \end{eqnarray}
    Let $\mfk g(n)\subsetneq \mfk{dg}(n)$ denote the 
    Lie subalgebra of $\mfk{dg}(n)$ of 
    quadratic monomials, with dimension $2n^2-n$: 
    \begin{eqnarray}
        \mfk{g}(n) = \{
            a \gamma_j \gamma_k | a\in \C, \{j, k\}\subset[2n], j\neq k
        \}. 
    \end{eqnarray}
    The image of $\mfk g(n), \mfk{dg}(n)$ under the exponential map 
    are the group $G(n)$ of even Gaussian unitaries and 
    $DG(n)\supsetneq G(n)$ of displaced Gaussian unitaries, respectively. 
\end{definition}
The following isomorphism can be verified by direct computation 
of the Lie bracket. 
\begin{proposition}[even Gaussian algebra isomorphism]
    The even Gaussian Lie algebra $\mfk g(n)$ is isomophic to the algebra of 
    antisymmetric matrices $\mfk{so}(2n)$ under 
    $\varphi:\mfk g(n)\to \mfk{so}(2n)$ defined by: 
    \leqalign{def:soBasis}{ 
        \varphi(\gamma_a \gamma_b) = 2s_{ab}, \quad 
        s_{ab} = |a\ra \la b| - |b\ra \la a|, \quad a\neq b. 
    }
\end{proposition}
An embedding is an injective homomorphism of algebras. 
The next embedding introduced in~\cite{knill2001fermionic} 
is central to the constructions we use in this work.
\begin{definition}[Clifford algebra embedding]
    \label{def:clifAlgEmbedding}
    The following map $\phi:\Cl_{2n}\to \Cl_{2n+1}$ is an embedding of the 
    Clifford algebra, extended multiplicatively from the generators 
    according to 
    \begin{eqnarray}
        \phi(\gamma_j) = i\gamma_j\gamma_{2n+1}. 
    \end{eqnarray}
    For an arbitrary basis element $\gamma_J\in \Cl_{2n}$ 
    of degree $|J|=m$, we obtain 
    \leqalign{eq:evenEmbeddingFormula}{ 
        \phi(\gamma_J) = \begin{cases}
            i \gamma_J \gamma_{2n+1} & m\text{ odd,} \\ 
            \gamma_J & m\text{ even}
        \end{cases}. 
    }
\end{definition}
Since $\phi$ acts as an injective homomorphism by extending 
multiplicatively from its action on generators, 
it respects the Lie bracket. 
Consequently, $\phi$, when restricted to $\mfk{dg}(n)$, 
provides a Lie algebra embedding $\mfk{dg}(n) \to \mfk {g}(n+1)$. 
Composing $\varphi$ and $\phi$ yields the 
raw extended covariance matrix (equation~\ref{def:rawExtCovMain}), 
and the following property explains the ubiquity of the extended 
covariance matrix in the theory 
of displaced Gaussian operators. 
\begin{proposition}[$\mfk{dg}(n)$ to $\mfk{so}(2n+1)$ embedding]
    \label{prp:dispAntiSymmetryEmbedding}
    The composition 
    $2\bar \Sigma = \varphi \circ \phi:\mfk{dg}(n)\to \mfk{so}(2n+1)$ 
    satisfying 
    \leqalign{eq:dispAntiSymmetryEmbedding}{
        \quad O = \df 1 2 \gamma^T M \gamma + d^T \gamma \in \mfk{dg}(n) 
        \implies 2\bar \Sigma(O)
        = 2 \begin{bmatrix}
            M & id \\ -id^T & 0 
        \end{bmatrix} \in \mfk{so}(2n+1, \C) 
    }
    is an isomorphism between the antisymmetric Lie algebra
    and the displaced Gaussian algebra. 
\end{proposition}
Observe that the quadratic terms 
$2\bar{\Sigma}(\gamma_j\gamma_k)=2s_{jk}$, 
where $1\leq j<k\leq 2n$, generate rotations between the 
$j, k$-th  subspaces, while the linear terms 
$2\bar{\Sigma}(\gamma_j) = 2i s_{j(2n+1)}$ induce rotations 
between the $j$-th and $(2n+1)$-th 
subspaces.

\section{Wick's formula}
\label{app:extWick}

\begin{definition}[Pfaffian]
    Given an antisymmetric matrix $M\in \mfk{so}(2n, \C)$, 
    the Pfaffian~\cite{jaffe1989pfaffians} of $M$ is defined by
    \begin{eqnarray}
        \Pf(M) = \df 1 {2^n n!} \sum_{\sigma \in S_{2n})} 
        \prod_{l=1}^n M_{\sigma(2l-1), \sigma(2l)}. 
    \end{eqnarray}
    Given an even multi-index $J\subset [n]$ with $|J|=2m$, 
    we denote by $M_{|J}\subset \mfk{so}(|J|, \C)$ the restriction 
    of $M$ to the subspaces indexed by $J$. The component formula 
    and its Pfaffian are, correspondingly 
    \leqalign{eq:restPfaffian}{ 
        \left(M_{|J}\right)_{jk} = M_{J_jJ_k}\implies 
        \Pf\left(M_{|J}\right) = \df 1 {2^m m!} \sum_{\sigma \in S_{2m})} 
        \prod_{l=1}^m M_{J_{\sigma(2l-1)}J_{\sigma(2l)}}. 
    }
\end{definition}

A displaced Gaussian state is completely specified by its extended 
covariance. The following result 
slightly generalizes Wick's formula for 
even Gaussian states~\cite{surace2022fermionic} to the nonzero-mean case. 
It is the anti-commuting counterpart 
of the classical Isserlis-Wick theorem 
for computing the higher-order moments of 
multivariate Gaussian distributions. 

\begin{proposition}[restatement of proposition~\ref{prp:extWickMain}]
    \label{prp:extWick}
    The $J$th moment of a displaced Gaussian state $\rho$, 
    is expressible in terms of 
    $\tilde \Sigma(\rho)$ using the following formula: 
    \leqalign{eq:extWickDefs}{
        \rho_J = \alpha_{|J|} \Pf\left[\tilde \Sigma(\rho)_{|\tilde J}\right], 
        \quad 
        \tilde J = \begin{cases}
            J & |J| \text{ even} \\ 
            J\cup \{2n+1\} & |J| \text{ odd}
        \end{cases}, \quad 
        \alpha_{|J|} = (-i)^{|J|\,\operatorname*{mod}\, 2}
    }
    Here $\tilde \Sigma(\rho)_{|\tilde J}$ 
    denotes the restriction of $\tilde \Sigma(\rho)$ to the subspaces 
    indexed by $\tilde J$ as in equation~\ref{eq:restPfaffian}. 
    Equivalently, 
    \leqalign{eq:wickExpansion}{
        \rho = \df 1 {2^n} \sum_J \rho_J \gamma_J 
        = \df 1 {2^n} \sum_J  \alpha_{|J|}  \Pf\left[
            \tilde \Sigma(\rho)_{|\tilde J}\right]\gamma_J. 
    }
    \begin{proof}
        We begin by denoting the covariance and mean vectors as 
        $\Sigma=\Sigma(\rho)$ and $\mu = \mu(\rho)$, respectively. 
        Expanding the Fourier transform definition, 
        we have: 
        \leqalign{eq:wickIntermediate1}{
            \Xi_\rho(\theta) 
            &= \sum_J \rho_J \theta_J 
            = \sum_{k=1}^{2n} \df 1 {k!} 
            \left(\df 1 2 \theta^T \Sigma \theta + \mu^T\theta\right)^k \\ 
            &= \sum_{k=1}^{n} \df 1 {k!} \left[
                \left(\df 1 2 \theta^T \Sigma \theta\right)^k + 
                k \left(\df 1 2 \theta^T \Sigma \theta
                \right)^{k-1} \left(\mu^T\theta\right)
            \right]. 
        }
        This last equality holds because only two terms in the binomial 
        expansion of each $k$-th power are non-zero. In particular, 
        any term containing more than one factor of $(\mu^T \theta)$ 
        vanishes since $(\mu^T \theta)^2=0$. Hence, we obtain:
        \leqalign{eq:wickIntermediate2}{
            \left(\df 1 2 \theta^T \Sigma \theta 
            + \mu^T\theta\right)^k 
            = \left(\df 1 2 \theta^T \Sigma \theta\right)^k + 
            k \left(\df 1 2 \theta^T \Sigma \theta\right)^{k-1} 
            \left(\mu^T\theta\right). 
        }
        Now, consider the case where $|J|=2m$ (i.e., $J$ has even degree). 
        According to equation~\ref{eq:wickIntermediate1}, 
        the coefficient $\rho_J$ is found as the term in front of 
        $\theta_J$ in the expansion:
        \malign{ 
            \df 1 {2^m m!} \left(\sum_{j, k=1}^{2n} 
            \Sigma_{jk} \theta_j \theta_k\right)^m 
            &= \sum_{|K|=2m} \df 1 {2^m m!} \sum_{\sigma \in S_{2m}} 
            \prod_{l=1}^m \Sigma_{K_{\sigma(2l-1)}K_{\sigma(2l)}}\theta_K \\
            &= \sum_{|K|=2m} \Pf\left(\Sigma_{|K}\right) \theta_K 
            = \sum_{|K|=2m} \Pf\left[\tilde \Sigma(\rho)_{|K}\right] \theta_K. 
        }
        Here, the first equality follows from the combinatorics, 
        the second from equation~\ref{eq:restPfaffian}, 
        and the final line follows by applying 
        equation~\ref{eq:extWickDefs}, noting that the sum is 
        over multi-indices $K \subset [2n]$ of size $|J|=2m$. 
        This completes the proof for even-degree $J$.
        For odd $|J|=2m-1$, we focus on the odd term in 
        equation~\ref{eq:wickIntermediate1}, 
        which expands to $\frac{1}{(m-1)!} \left(\frac{1}{2} 
        \theta^T \Sigma \theta\right)^{m-1}(d^T\theta)$. 
        Here, $\rho_J$ corresponds to the term in front of 
        $\theta_J$ in this expansion. 
        Since $\mu^T\theta$ is the sole odd term in this product, 
        it commutes with every other term. This 
        justifies the introduction of an additional 
        Grassmann variable. Writing 
        $\tilde \theta = (\theta_1, \dots, \theta_{2n}, \theta_{2n+1})$, 
        $\rho_J$ corresponds to the coefficient of 
        $\theta_{\tilde J}$ in:
        \malign{ 
            \frac 1 {(m-1)!} \left(\frac 1 2 \theta^T \Sigma \theta\right)^{m-1}(\mu^T\theta)
            &\mapsto \frac 1 {(m-1)!} \left(\frac 1 2 \tilde \theta^T \begin{bmatrix}
                \Sigma & 0_{2n\times 1} \\ 0_{1\times 2n} & 0 
            \end{bmatrix} \tilde \theta\right)^{m-1} \left(\df 1 2 \theta^T \mu \cdot \theta_{2n+1} 
            - \df 1 2 \theta_{2n+1}\cdot \mu^T \theta\right) \\ 
            &= (-i) \frac 1 {(m-1)!} \left(\frac 1 2 \tilde \theta^T \begin{bmatrix}
                \Sigma & 0_{2n\times 1} \\ 0_{1\times 2n} & 0 
            \end{bmatrix} \tilde \theta\right)^{m-1} \left(\df 1 2 \tilde \theta^T \begin{bmatrix}
                0_{2n\times 2n} & i\mu \\ -i\mu^T & 0 
            \end{bmatrix}\tilde \theta\right)\\ 
            &= (-i) \sum_{|K|=2m, 2n+1\in K} \Pf\left[\tilde \Sigma(\rho)_{|K}\right] \tilde \theta_K. 
        }
        This completes the proof that for odd-degree $J$ that 
        $\rho_J = (-i) \Pf\left[\tilde \Sigma(\rho)_{|\tilde J}\right]$. 
    \end{proof}
\end{proposition}

\section{Conjugate action of $DG(n)$}
\label{app:unitaryAction}
The first step in understanding displaced Gaussian operators 
is characterizing the conjugate action of $DG(n)$. 
We begin by identifying this conjugate action on $\mfk{dg}(n)\subset \Cl_{2n}$, 
which turns out to be a rotation of the quadratic polynomials. 
We next show that this property tensorizes properly, 
yielding a compact form for conjugation on $\Cl_{2n}$. 
Finally, we derive the conjugate action of displaced Gaussian 
unitaries on displaced Gaussian states. 

\begin{lemma}[displaced Gaussian action on $\mfk{dg}(n)$]
    \label{lem:dispUnitaryLieAlgebraEffect}
    Given $O\in \mfk{dg}(n)$ and displaced Gaussian unitary $U\in DG(n)$ 
    \begin{eqnarray}
        O = \df 1 2 \gamma^T M \gamma + v^T\gamma, \quad 
        U = e^B, \quad B = \df 1 2 \gamma^T h \gamma + i d^T\gamma \in \mfk{dg}(N). 
    \end{eqnarray}
    The conjugate action of $U$ on $O$ satisfies, for $\log U=B \in \mfk{dg}(n)$, 
    \begin{eqnarray}
        \tilde \Sigma(UOU^\dag) 
        = R \, \tilde \Sigma(O) R^T, \quad R = e^{2\bar \Sigma(\log U)}.  
    \end{eqnarray}
    We can expand the definitions to obtain the explicit formula: 
    \leqalign{}{
        UOU^\dag = \df 1 2 \gamma^T A \gamma + u^T\gamma 
        \quad \text{where}\quad 
        \begin{bmatrix}
            A & iu \\ -iu^T & 0 
        \end{bmatrix} = R \begin{bmatrix}
            M & iv \\ -iv^T & 0 
        \end{bmatrix} R^T, \quad 
        R = \exp \left(2 \begin{bmatrix}
            h & -d \\ d^T & 0 
        \end{bmatrix}\right). 
    }
    Setting $M=0$ and $d=0$ specializes 
    to the known action of even Gaussian unitaries 
    (\cite[theorem 3]{jozsa2008matchgates}). 
    \begin{proof}
        Using the Baker-Campbell-Hausdorff formula, 
        conjugation is determined by the Lie bracket on $\mfk{dg}_n$: 
        \malign{
            U O U^\dag 
            &= e^B O e^{-B}
            = O + [B, O] + \frac{1}{2!} [B, [B, O]] 
            + \frac{1}{3!} [B, [B, [B, O]]] + \dots \\ 
            &= \sum_{n=0}^{\infty} \frac{1}{n!} \text{ad}_{B}^n(O), \quad 
            \text{ad}_{B}(O) = [B, O]. 
        }
        Using the compatibility of the adjoint map 
        with $2\bar \Sigma$ (proposition~\ref{prp:dispAntiSymmetryEmbedding}) yields 
        \leqalign{eq:adjointEquation}{
            2\bar \Sigma(UOU^\dag)
            &= \sum_{n=0}^{\infty} 
                \frac{1}{n!} \text{ad}_{2\bar \Sigma(B)}^n[2\bar \Sigma(O)]
            = e^{2\bar\Sigma(B)} 2\bar \Sigma(O) e^{-2\bar \Sigma(B)}. 
        }
    \end{proof}
\end{lemma}

\begin{lemma}
    \label{lem:evenRotation}
    Given $U=\exp\left(\gamma^T h \gamma / 2\right)\in G(n)$, 
    its conjugate action 
    affects a Majorana basis element $\gamma_J$ 
    by antisymmetrized rotation within the degree-$m$ subspace: 
    \leqalign{eq:evenGaussianBasisAction}{
        U\gamma_J U^\dag 
        = \sum_{|K|=|J|} R_{\hat K, J} \gamma_K, \quad R = e^{2h}. 
    }
    Here the sum is over all size-$m$ multi-indices (sorted subsets 
    of $[2n]$), and $R_{\hat K, J}$ denotes the following 
    antisymmetrized product defined for $|J| = |K| = m$
    (for example, $R_{\widehat{[a, b]}, [c, d]} = R_{ac}R_{bd} - R_{bc}R_{ad}$): 
    \leqalign{eq:antisymRotation}{
        R_{\hat K, J} = \sum_{\sigma \in S_m} \mrm{sgn}(\sigma) 
        \prod_{i=1}^{|J|} R_{K_{\sigma(i)}, J_i}. 
    }
    \begin{proof}
        Substituting $M, d=0$ in 
        lemma~\ref{lem:dispUnitaryLieAlgebraEffect}
        shows that $U\in G(n)$ affects a rotation of the generators: 
        \leqalign{eq:evenGaussianGeneratorAction}{
            U \gamma_j U^\dag 
            = \sum_{k=1}^{2n} R_{kj} \gamma_k, 
            \quad R = e^{2h}
            , \quad \forall j\in [2n]. 
        }
        To see how equation~\ref{eq:evenGaussianGeneratorAction} 
        implies~\ref{eq:evenGaussianBasisAction}, expand 
        \malign{
            U\gamma_J U^\dag 
            &= \prod_{j=1}^{m} U\gamma_{J_j} U^\dag 
            = \prod_{j=1}^{m} \left(\sum_{k=1}^{2n} R_{k, J_j}\gamma_k\right) \\ 
            &= \left(
                \sum_{K_1=1}^{2n} R_{K_1, J_1} \gamma_{K_1}
            \right) \cdots \left(
                \sum_{K_m=1}^{2n} R_{K_m, J_m} \gamma_{K_1}
            \right) \\ 
            &= \sum_{K_1=1}^{2n}\cdots \sum_{K_m=1}^{2n} 
            \left(\prod_{i=1}^m R_{K_i, J_i}\right) \gamma_K 
            = \sum_{|K|=m} R_{\hat K, J}\gamma_K. 
        }
        In the last equality, the $(K_1, \cdots K_m)$ cannot contain duplicates 
        due to the orthogonality of $R$ and the elements of $J$ being distinct: 
        it is insightful to investigate the $m=2$ case, 
        in which case the duplicate terms sum to 
        \[
            \sum_{k=1}^{2n} R_{k, J_1}R_{k, J_2} \gamma_k^2 = 2I (R^TR)_{J_1, J_2} = 0. 
        \]
        Consequently, the sum over distinct $(K_1, \cdots, K_m)$ can be decomposed 
        into a sum over sorted $(K_1, \cdots, K_m)$ together with a sum over all 
        permutations over $S_m$, yielding the desired equation. 
    \end{proof}
\end{lemma}

We are now ready to prove the conjugate action of $DG(n)$ on all 
of $\Cl_{2n}$. 
Recall the definition of $\tilde J$ in~\ref{prp:extWick} 
and the antisymmetrized product $R_{\hat L, J}$ in 
equation~\ref{eq:antisymRotation}. The following result 
shows that conjugating $\gamma_J$ by $U$ is equivalent to 
adjoining an additional $\gamma_{2n+1}$ using $\phi$, 
rotating the even basis $\gamma_{\tilde J}$ 
using the antisymmetrized $R$ as in equation~\ref{eq:antisymRotation}, 
then deleting the adjoined mode. 

\begin{theorem}[displaced Gaussian action on $\Cl_{2n}$]
    \label{thm:dispUnitaryWholeAlgebraEffect}
    Given $U\in DG(n)\subset \Cl_{2n}$ effecting $R\in SO(2n+1)$ 
    per lemma~\ref{lem:dispUnitaryLieAlgebraEffect}. 
    its conjugate action affects a Majorana monomial 
    $\gamma_J \in \Cl_{2n}$ by the following equation: 
    \begin{eqnarray}
        U\gamma_J U^\dag = \sum_{|\tilde K|=|\tilde J|} 
        R_{\widehat{\tilde K}, \tilde J} \gamma_K
    \end{eqnarray}
    where the sum is over all multi-indices $K\subset [2n]$ 
    such that $|\tilde K| = |\tilde J|$. 
    \begin{proof}
        Recall equation~\ref{eq:evenEmbeddingFormula} and let 
        $\phi(\gamma_J) = \alpha_{|J|} \gamma_{\tilde J}$, 
        where $\alpha_{|J|} = i^{|J|\mod 2}$. 
        Let $U=e^B\in DG(n)$ with $B$ quadratic and 
        $\tilde U = \phi(e^B) = e^{\phi(B)} \in G(n+1)$, we obtain 
        \malign{
            \phi(U\gamma_J U^\dag) 
            &= \phi \left[\sum_{n=0}^\infty 
            \df 1 {n!} \mrm{ad}^n_B(\gamma_J)
            \right] 
            = \sum_{n=0}^\infty \df 1 {n!} 
            \mrm{ad}^n_{\phi(B)}[\phi(\gamma_J)] 
            = e^{\phi(B)} \phi(\gamma_J) e^{-\phi(B)}
            =\alpha_{|J|} \tilde U \gamma_{\tilde J} \tilde U^\dag. 
        }
        Let $R=e^{2\bar \Sigma(\log U)}\in SO(2n+1)$ be the rotation corresponding 
        to $U$ in lemma~\ref{lem:dispUnitaryLieAlgebraEffect}, and 
        $\tilde R\in SO(2n+2)$ be the rotation corresponding to the even 
        Gaussian unitary $\tilde U$. They are related by 
        \leqalign{eq:dispRotation}{
            \tilde R = \exp \left(2 \begin{bmatrix}
                h & -d \\ d^T & 0 
            \end{bmatrix}\oplus (0)\right) 
            = R \oplus (1) \in SO(2n+2). 
        }
        Invoking lemma~\ref{lem:evenRotation} on the even conjugate action 
        $\tilde U \gamma_{\tilde J}\tilde U^\dag$ yields 
        \malign{ 
            \phi(U\gamma_J U^\dag)
            = \alpha_{|J|} \tilde U \gamma_{\tilde J} \tilde U^\dag
            = \alpha_{|J|}
            \sum_{|\tilde K|=|\tilde J|, \tilde K\subset [2n+2]} 
            \tilde R_{\widehat{\tilde K}, \tilde J} \gamma_{\tilde K}
            = \alpha_{|J|} 
            \sum_{|\tilde K|=|\tilde J|, \tilde K\subset [2n+1]} 
            R_{\widehat{\tilde K}, \tilde J} \gamma_{\tilde K}
        }
        The last equality holds in light of equation~\ref{eq:dispRotation}: 
        we can ignore the $(2n+2)$-th subspace since the conjugated term $\phi(\gamma_J)$ 
        does not contain $\gamma_{2n+2}$ and $\tilde R$ acts trivially on this subspace. 
        Noting that $\alpha_{|J|} = \alpha_{|K|}$  
        for $|\tilde K|=|\tilde J|$ (since $\alpha_J$ only depends on the 
        degree), apply $\phi^{-1}$ to both sides yields the desired relation. 
    \end{proof}
\end{theorem}


To analyze how a displaced Gaussian unitary acts on Gaussian states, 
we need a lemma concerning the Pfaffian. 
\begin{lemma}
    \label{lem:PfaffianRotation}
    Given an antisymmetric covariance matrix 
    $A\in \mfk{so}(2n, \C)$ and a rotation matrix $R\in SO(2n, \R)$, 
    the following equation holds for all $m=1, \dots, n$: 
    \leqalign{eq:pfFormula}{
        \sum_{|K|=|J|=2m} 
        \Pf\left(A_{|J}\right) R_{\hat K, J} \gamma_K 
        = \sum_{|J|=2m} \Pf\left[(RAR^T)_{|J}\right] \gamma_J
    }
    with $R_{\hat K, J}$ as defined in equation~\ref{eq:antisymRotation};  
    the left-hand sum is over even multi-indices $J, K\subset [2n]$ with 
    $|J|=|K|=2m$. 
    \begin{proof}
        It is known that (\cite{jozsa2008matchgates}, theorem 3), 
        given an even Gaussian state $\rho$ and 
        $U = e^B\in G(n)$, the state $U\rho U^\dag$ remains even Gaussian 
        and has the following covariance matrix: 
        \[ 
            \Sigma(U\rho U^\dag) = R\, \Sigma(\rho) R^T, \quad 
            R = e^{2\bar \Sigma(B)}. 
        \] 
        Let $A = \Sigma(\rho)$, 
        apply Wick's formula~\ref{prp:extWick} to $U\rho U^\dag$ 
        to obtain the right-hand side of equation~\ref{eq:pfFormula}: 
        \malign{
            U\rho U^\dag = \df 1 {2^n} \sum_J \Pf\left[
                \Sigma(U\rho U^\dag)_{|J}
            \right] \gamma_J
            = \df 1 {2^n} \sum_{J} \Pf\left[
                (RAR^T)_{|J}
            \right] \gamma_J. 
        }
        To obtain the left-hand side in equation~\ref{eq:pfFormula}, 
        apply lemma~\ref{lem:evenRotation} to each term in the Wick's formula 
        expansion of $\rho$: 
        \malign{
            U\rho U^\dag &= \df 1 {2^n} \sum_J \Pf(A_{|J}) U\gamma_JU^\dag 
            = \df 1 {2^n} \sum_J \Pf(A_{|J}) 
            \sum_{|K|=|J|} R_{\hat K, J} \gamma_K. 
        }
        To conclude the proof, the two expressions for $U\rho U^\dag$ 
        implies the following equation, which must also hold for each 
        degree $|J|=|K|=2m$ separately since addition does not change 
        the degree of a term: 
        \[
            \sum_{J} \Pf\left[
                (RAR^T)_{|J}
            \right] \gamma_J = \sum_J \Pf(A_{|J}) 
            \sum_{|K|=|J|} R_{\hat K, J} \gamma_K. 
        \] 
    \end{proof}
\end{lemma}

\begin{theorem}[restatement of theorem~\ref{thm:dgUnitaryActionGaussian}]
    \label{thm:dgStateAction}
    Applying a displaced Gaussian unitary $U\in G(n)$ 
    to  $\rho \in \dgauss(n)$ 
    results in a displaced Gaussian state with the extended 
    covariance matrix 
    \leqalign{eq:dgStateAction}{ 
        \tilde \Sigma(U\rho U^\dag) = R \, \tilde \Sigma(\rho) R^T\text{ where }
        R = e^{2\bar \Sigma(\log U)}.  
    }
    Here $\log U\in \mfk{dg}(n)$ is proportional to the quadratic Hamiltonian 
    generating $U$, and $\bar \Sigma$ is defined in equation~\ref{def:rawExtCovMain}. 
    \begin{proof}
        The first equality below follows from applying 
        Wick's formula~\ref{prp:extWick} to $U\rho U^\dag$. 
        It remains to demonstrate the second inequality in 
        \malign{
            \df 1 {2^n} \sum_{J\subset [2n]} \alpha_J 
            \Pf \left[\tilde \Sigma(U\rho U^\dag)_{|\tilde J}\right] \gamma_{J}
            = U\rho U^\dag 
            = \df 1 {2^n} \sum_{J\subset [2n]} \alpha_J 
            \Pf \left[(R \tilde \Sigma(\rho) R^T)_{|\tilde J}\right] \gamma_{J}. 
        }
        To do so, apply theorem~\ref{thm:dispUnitaryWholeAlgebraEffect} 
        to each term in the Wick expansion of $\rho$ to obtain 
        \malign{
            U\rho U^\dag &= \df 1 {2^n} 
            \sum_J \alpha_J \Pf[\Sigma(\rho)_{\tilde J}] U\gamma_JU^\dag  
            = \df 1 {2^n} 
            \sum_J \alpha_J \Pf[\Sigma(\rho)_{\tilde J}] 
            \sum_{|\tilde K|=|\tilde J|} 
            R_{\widehat {\tilde K}, \tilde J} \gamma_{K}. 
        }
        Noting that $\alpha_J$ only depends on $|J|$, 
        invoking lemma~\ref{lem:PfaffianRotation} 
        with $(\tilde J, \tilde K)$ in place of $(J, K)$ 
        concludes the proof: 
        \malign{
            U\rho U^\dag &= \df 1 {2^n}  
            \sum_J \sum_{|\tilde K|=|\tilde J|}
            \alpha_J \Pf[\Sigma(\rho)_{\tilde J}] 
            R_{\widehat {\tilde K}, \tilde J} \gamma_{K}
            = \df 1 {2^n} \sum_J \alpha_J \Pf \left[
                (R\, \Sigma(\rho) R^T)_{|J}
            \right]\gamma_J. 
        }
    \end{proof}
\end{theorem}

\section{Characterization of displaced Gaussian states}
\label{app:dgStateCharacterization}
In this section, we unify three reasonable definitions of 
displaced Gaussian states and show that they are equivalent.
This is done by first establishing this characterization for
a special diagonalizable subclass of displaced Gaussian states, 
then extending this result by showing that the full set of displaced 
Gaussian states is the orbit of the diagonalizable 
Gaussian states under $DG(n)$. 

\begin{definition}[diagonalized Gaussian states]
    \label{def:diagGaussianState}
    A displaced Gaussian state $\rho_D\in \Cl_{2n}$ is a diagonalized Gaussian 
    state if its extended covariance matrix is block-diagonalized: 
    \malign{
        \tilde \Sigma(\rho_D) = i \left(
            \bigoplus_{j=1}^n \begin{bmatrix}
                0 & \lambda_j \\ -\lambda_j & 0 
            \end{bmatrix}\right) \oplus (0) \in \mfk{so}(2n+1, \C). 
    }
\end{definition}

\begin{lemma}[diagonalizable Gaussians are separable computational basis mixtures]
    \label{lem:diagCircuitRep}
    With $\lambda_1$ as in definition~\ref{def:diagGaussianState}: 
    \malign{
        \rho_D = \bigotimes_{j=1}^n \left(\df{1 + \lambda_j}{2}|0\ra \la 0| 
        + \df{1 - \lambda_j}{2} |1\ra \la 1|\right). 
    }
    \begin{proof}
        Expand the Grassmann expression of $\rho_D(\theta)$ 
        using the covariance matrix:  
        \leqalign{}{
            \rho_D(\theta) 
            &= \exp \left[\df 1 2 \theta^T \Sigma(\rho_D) \theta\right]
            = \exp \left(\sum_{j=1}^n \lambda_j \theta_{2j-1}\theta_{2j}\right)  
            = \bigotimes_{j=1}^n \exp(i \lambda_j \theta_1\theta_2) 
            = \bigotimes_{j=1}^n 1 + i\lambda_j \theta_1\theta_2
        }
        Replacing $\theta_j\mapsto \gamma_j$ and recognizing $i\gamma_1\gamma_2 = Z$ 
        yields the desired equation 
        \leqalign{eq:diagonalizedGaussian}{ 
            \rho_D 
            &= \df 1 {2^n} \bigotimes_{j=1}^n 1 + i\lambda_j \gamma_1\gamma_2 
            =  \bigotimes_{j=1}^n \df{1 + i \lambda_j\gamma_1\gamma_2}{2}
            = \bigotimes_{j=1}^n \df{1 + \lambda_j Z}{2} 
            = \bigotimes_{j=1}^n \left(\df{1 + \lambda_j}{2}|0\ra \la 0| 
            + \df{1 - \lambda_j}{2} |1\ra \la 1|\right)
        }
    \end{proof}
\end{lemma}

\begin{lemma}[diagonalizable Gaussian states as thermal states]
    \label{lem:diagThermalRep}
    Every diagonalizable Gaussian state $\rho_D\in \Cl_{2n}$ as 
    in definition~\ref{def:diagGaussianState} is the 
    thermal (ground) state of a quadratic Hamiltonian $H\in \mfk{g}(n)$: 
    \begin{eqnarray}
        \rho_D = \df{e^{H}}{\tr(e^H)}, \quad 
        H = \df i 2 \gamma^T h \gamma, \quad 
        h = \bigoplus_{j=1}^n \begin{bmatrix}
            0 & \mrm{tanh}^{-1}\lambda_j \\ 
            -\mrm{tanh}^{-1}\lambda_j & 0 
        \end{bmatrix}
    \end{eqnarray}
    \begin{proof}
        Using lemma~\ref{lem:diagCircuitRep}: 
        w.l.o.g. we can consider a single-qubit with 
        $h_1 = \arctan(\lambda_j)(|0\ra \la 1| - |1\ra \la 0|)$, then 
        \malign{
            H_1 = \df i 2 \begin{pmatrix}
                \gamma_1 \\ \gamma_2 
            \end{pmatrix}^T \begin{pmatrix}
                0 & \mrm{tanh}^{-1}\lambda \\ 
                -\mrm{tanh}^{-1}\lambda & 0
            \end{pmatrix} \begin{pmatrix}
                \gamma_1 \\ \gamma_2 
            \end{pmatrix} 
            = (\mrm{tanh}^{-1}\lambda)Z 
        } 
        The thermal state of $H_1$ is 
        $\frac 1 2 \left[
            (1+\lambda)|0\ra \la 0| + (1 - \lambda)|1\ra \la 1|
        \right]$. Pure states with $\lambda = \pm 1$ 
        are defined as the limit. 
    \end{proof}
\end{lemma}

\begin{lemma}[diagonalizability of displaced Gaussian states]
    \label{lem:diagonalizability}
    For every displaced Gaussian state $\rho\in \Cl_{2n}$, 
    there exists $U\in DG(n)$ such that $U\rho U^\dag$ is 
    a diagonalized Gaussian state. If $\rho$ is an even Gaussian 
    state, then $U$ will be even. 
    \begin{proof}
        By standard matrix theory~\cite{zumino1962normal}, for every 
        $A \in \mfk{so}(m, \R)$ there exists 
        a rotation $R \in SO(m, \R)$ such that $RAR^T$ is block-diagonal. 
        lemma~\ref{lem:dispUnitaryLieAlgebraEffect} shows 
        that the displaced Gaussian unitaries $DG(n)\cong SO(2n+1)$ 
        and that rotations acting trivially on the last subspace 
        are generated by the even Gaussian unitaries. 
        Given a displaced Gaussian state $\rho$, 
        $\tilde{\Sigma}(\rho)$ is antisymmetric 
        thus block-diagonalizable by the conjugate action of some $U\in DG(n)$
        per theorem~\ref{thm:dgStateAction}. 
        If $\rho$ is even, $\tilde \Sigma(\rho)$ is trivial 
        on the last subspace thus block-diagonalizable by $U\in G(n)$. 
    \end{proof}
\end{lemma}

\begin{theorem}[restatement of theorem~\ref{thm:unifiedCharacterization}]
    \label{thm:dgStateCharacterization}
    A $n$-qubit state $\rho \in \dgauss(n)$ 
    iff any of the following holds:
    \begin{enumerate}
        \item \emph{Thermal state definition}: 
        $\rho$ is the thermal state of 
        some quadratic Hamiltonian $H\in \mfk{dg}(n)$, 
        \begin{eqnarray}
            \rho = \df{e^{H}}{\tr(e^H)}, \quad 
            H = \df i 2 \gamma^T h \gamma + d^T\gamma . 
        \end{eqnarray}
        If $\rho$ is pure, then $\rho$ is the ground state of 
        some quadratic Hamiltonian. 
        \item \emph{Circuit definition}: 
        $\rho$ is the result of a displaced Gaussian unitary 
        acting on a separable computational basis state: 
        \begin{eqnarray}
            \rho = U_G \rho_D U_G^\dag, \quad U_G\in DG(n), \quad 
            \rho_D \text{ diagonal Gaussian.}
        \end{eqnarray}
        \item \emph{Fourier definition}: 
        The Fourier transform of $\rho$ admits a Gaussian expression: 
        \begin{eqnarray}
            \Xi_\rho(\theta) = 
            \exp \left(\df i 2 \theta^T M \theta + d^T\theta\right)
        \end{eqnarray}
    \end{enumerate}
    \begin{proof}
        (3$\implies$1 and 2): given $\rho\in \Cl_{2n}$ with 
        a Gaussian Fourier expression, 
        lemma~\ref{lem:diagonalizability} provides 
        $U_G\in DG(n)$ such that $U\rho U^\dag=\rho_D$ is a 
        diagonalized Gaussian state, from 
        which (2) immediately follows. 
        lemma~\ref{lem:diagThermalRep} implies that 
        $\rho_D \propto e^{H_D}$ for some quadratic $H_D$, 
        thus $\rho = U^\dag \rho U \propto e^{U_G^\dag H_D U_G}$ 
        is the thermal state of the Hamiltonian $U_G^\dag H_D U_G$ 
        which remains quadratic by lemma~\ref{lem:dispUnitaryLieAlgebraEffect}. 
        The implication (2$\implies$3) follows from theorem~\ref{thm:dgStateAction}, 
        and (1$\implies$3) follows from block-diagonalizing $\tilde \Sigma(H)$ 
        using $U_G$ such that $U_G\rho U_G^\dag = \rho_D$, 
        identifying the Fourier expression of $\rho_D$, 
        then applying theorem~\ref{thm:dgStateAction} to 
        $U_G^\dag \rho_D U_G = \rho$. 
    \end{proof}
\end{theorem}



\section{Classical simulation of displaced Gaussian circuits}
\label{app:classicalSimulation}
It has been shown that every $U\in G(n)$ has a decomposition 
into $O(n^3)$ local n.n. matchgates which act on at most two consecutive 
lines \cite[theorem 5]{jozsa2008matchgates}. 
The following result extends this decomposition by 
applying the same technique 
slightly adopted to displaced Gaussian states. 
\begin{theorem}[restatement of theorem~\ref{thm:dispDecomposition}]
    Every $U\in DG(n)$ can be decomposed into the product of 
    $O(n^3)$ gates which are either matchgates 
    or single-qubit gates on the initial 
    line of the Jordan-Wigner transform. 

    \begin{proof}
        Proposition~\ref{prp:dispAntiSymmetryEmbedding} 
        shows that \( DG(n) \cong SO(2n+1) \). 
        Examining the embedding equation for the rotation \( R \),
        \[
            \log R = 2\bar \Sigma\left(\frac{1}{2} \gamma^T C \gamma + d^T\gamma\right)  
            = 2 \begin{bmatrix}
                C & id \\ -id^T & 0 
            \end{bmatrix},
        \]
        we see that the quadratic forms 
        \( \frac{1}{2} \gamma^T C \gamma \) generate 
        rotations within the first \( 2n \) subspaces. 
        Single-qubit gates can be generated using \( R_z \) 
        (a matchgate) along with \( R_x \), where \( R_x(\theta) 
        = \exp(i \theta X/2) = \exp(i\theta \gamma_1/2) \), 
        to create rotations between the first and \( (2n+1) \)-th 
        subspaces under \( 2\bar \Sigma \).
        An element of \( SO(2n+1) \) can be decomposed 
        into \( O(n^2) \) rotations between pairs of subspaces 
        via the method of Euler angles~\cite{hoffman1972generalization}. 
        Further, each rotation between two arbitrary subspaces can be 
        implemented with \( O(n) \) rotations between the 
        \( (j, k) \)-th subspaces (for \( 1 \leq j < k \leq 2n \)) 
        or between the first and \( (2n+1) \)-th subspaces. 
        Thus, the total decomposition requires \( O(n^3) \) gates, 
        each of which is either a matchgate or a single-qubit gate 
        on the initial line of the Jordan-Wigner transform.
    \end{proof}
    
\end{theorem}
The proposition below was essentially shown in~\cite[theorem 3]{brod2016efficient}. 
\begin{proposition}[restatement of proposition~\ref{prp:productsAreDg}]
    Every $n$-qubit product state is a displaced Gaussian state. 
\end{proposition}
\begin{proof}
    We use the fermionic swap unitary first 
    defined in~\cite{bravyi2002fermionic} 
    \leqalign{def:fswap}{
        S_{j\leftrightarrow k} = \exp \left[\df \pi 4 \left(
            \gamma_{2j-1}\gamma_{2k} - \gamma_{2j}\gamma_{2k-1} 
            - \gamma_{2j-1}\gamma_{2j} 
            - \gamma_{2k-1}\gamma_{2k}
        \right)\right] = S_{k\leftrightarrow j} 
        = -S_{j\leftrightarrow k}^\dag. 
    }
    The fermionic swap acts as a genuine swap if 
    any one of the input lines is even. 
    For every pure product state 
    \begin{eqnarray}
        |\psi\ra = |\psi_1\ra \cdots |\psi_n\ra 
        = (U_1\otimes \cdots \otimes U_n)|0\ra^{\otimes n}, 
    \end{eqnarray}
    we can implement it by computing $U_n|0\ra$ on the 
    initial line, swapping it through the intermediate computational 
    basis states to the $n$-th register, then do the same for
    $U_{n-1}|0\ra, \dots, U_1|0\ra$. 
    Since we're only using displaced Gaussian unitaries 
    consisting of single-qubit 
    states on the initial line and the fermionic swap, 
    the resulting pure product state is a displaced Gaussian state; 
    mixed product states follow by the same argument. 
\end{proof}

Given a displaced Gaussian state input, the action of displaced 
Gaussian unitaries can be efficiently simulated by theorem~\ref{thm:dgStateAction}. 
We next consider the simulation of measurements in the computational basis, 
which is can be done by computing the Grassmann integral of Gaussian operators. 
Using Gaussian integrals to facilitate simulation has been first considered in~\cite{bravyi2004lagrangian}, 
and earlier treatments of such integrals can be found 
in~\cite{lowdin1955quantum,soper1978construction}. We first briefly 
recount the definition of the Grassmann integral and known results for 
Grassmann integrals. 
\begin{definition}[Grassmann differentiation and integration]
    The partial derivative $\pd a: \G_n\to \G_n$ is the linear map 
    \begin{eqnarray}
        \pd a 1 = 0, \quad \pd a \theta_b = \delta_{ab}. 
    \end{eqnarray}
    Extended according to the \textit{Leibniz's rule} 
    $\pd a [\theta_b f(\theta)] = \delta_{ab} f(\theta) - \theta_b \pd a f(\theta)$. 
    Integration is equivalent to differentiation 
    \begin{eqnarray}
        \int d \theta_a \equiv \pd a, \quad 
        \int D\theta \equiv \int d\theta_n \cdots \int d\theta_2 \int d\theta_1. 
    \end{eqnarray} 
\end{definition}
Formally, $\pd a$ acts on a monomial $\theta_J$ by commuting $\theta_a$ 
to the left (if it exists in $\theta_J$), eliminating it, 
then keeping the remaining components; 
note that the image of $\pd a$ has no dependence on $\theta_a$. 
The order $\int D\theta$ is chosen such that $\int D\theta\, \theta_1\cdots \theta_n = 1$. 
Formally, $\int D\theta$ extracts the complex coefficient before the highest-order monomial 
in the polynomial expansion of an operator. Given antisymmetric $M\in \mfk{so}(2n, \C)$, 
we obtain (\cite{bravyi2004lagrangian}, equation 12): 
\leqalign{eq:evenGaussianInt}{
    \int D\eta\exp\left(\df 1 2 \eta^T M \eta \right)
    = \Pf(M). 
}
Given a nonsingular antisymmetric matrix $M\in \mfk{so}(2n, \C)$ and two sets
of Grassmann generators $\{\eta_j\}_{j=1}^{2n}, \{\theta_j\}_{j=1}^{2n}$
which anticommute with each other, i.e. $\{\theta_j, \eta_k\}=0$, then 
(\cite{bravyi2004lagrangian}, equation 13): 
\leqalign{eq:gaussianPartialInt}{ 
    \int D\theta \exp \left(\eta^T \theta + 
    \df 1 2 \theta^T M \theta\right)
    = \Pf(M) \exp \left(\df 1 2 \eta^T M^{-1}\eta \right). 
}
Given $X, Y\in \Cl_{2n}$ with Fourier transforms 
$\Xi_X(\theta), \Xi_Y(\eta)$ where 
$\theta, \mu$ are anti-commuting Grassmann generators, 
the trace of their product is~\cite[equation 15]{bravyi2004lagrangian}: 
\leqalign{eq:traceOverlap}{ 
    \tr(XY) = \left(-\df 1 2\right)^n \int 
    D\eta D\theta \,
        e^{\eta^T \theta} \, 
        \Xi_X(\eta) \, \Xi_Y(\theta). 
}

Using the formulas above,
we obtain the following result which underpins 
the efficient simulation of measurements. 
Note that it crucially requires one of the inputs to be even. 
\begin{lemma}[Gaussian overlap formula]
    \label{lem:dgOverlapFormula}
    Given a $n$-qubit displaced Gaussian state $\rho$ and 
    a $n$-qubit even Gaussian state $\sigma$, 
    let $A=\Sigma(\rho), B=\Sigma(\sigma)$. 
    If $B$ is invertible, then 
    \begin{eqnarray}
        \tr(\rho \sigma) = \df 1 {2^n} \sqrt{\det(I + AB)}. 
    \end{eqnarray}
    \begin{proof}
        Since $\sigma$ is even, the overlap $\tr(\rho \sigma)$ 
        only depends on the even coefficients of $\rho$. 
        Recalling Wick's formula~\ref{prp:extWick}, the even coefficients 
        of $\rho$ are unchanged if we set the mean to zero. 
        Thus without loss of generality expand 
        \[ 
            \Xi_\rho(\eta) = \exp \left(\df i 2 \eta^T A \eta\right), \quad 
            \Xi_\sigma(\theta) = \exp \left(\df i 2 \theta^T B \theta\right). 
        \] 
        First consider $B$ invertible,
        using the trace equation~\ref{eq:traceOverlap}
        and Gaussian integral equations~\ref{eq:evenGaussianInt},~\ref{eq:gaussianPartialInt} yields 
        \malign{
            (-2)^n \tr(\rho \sigma)
            &= \int D\eta D\theta \, e^{\eta^T \theta} X(\eta)Y(\theta) \\ 
            &= \int D\eta \,
            \exp \left(\df 1 2 \eta^T A \eta\right)
            \int D\theta \, e^{\eta^T\theta} \exp 
            \left(\df 1 2 \theta^T B \theta\right)  \\ 
            &= \Pf(B) \int D\eta \,
            \exp \left[\df 1 2 \eta^T (A+B^{-1}) \eta\right] 
            = \Pf(B)\Pf(A+B^{-1}). 
        } 
        Next note that $\Pf(A)^2 = \det A$ 
        and that $\tr(\rho \sigma)$since $\rho, \sigma$ are positive 
        operators, then 
        \malign{
            4^n \tr(\rho \sigma)^2 
            &= \left|
                \det(B) \det(A+B^{-1})
            \right|
            = |\det(AB + I)|. 
        }
        Finally, $AB$ is real since $A, B$ are each purely imaginary 
        matrices, and $\det(AB+I)$ is positive since $\|AB\|\leq \|A\|\|B\| \leq 1$ 
        in the operator norm. 
        This proves the result for inveritible $B$. 
        For the general case, let $B_\epsilon = B+\epsilon \tilde I$ 
        where $\tilde I=\bigoplus_{j=1}^n |0\ra \la 1|-|1\ra \la 0|$ 
        and $\epsilon \ll 1$ and take $\epsilon\to 0$. 
    \end{proof}
\end{lemma}

\begin{lemma}[restatement of lemma~\ref{lem:dgMeasurementMain}]
    Given a $n$-qubit displaced Gaussian state $\rho\in \Cl_{2n}$, 
    a subset $K\subset [n]$ of lines to measure 
    with $|K|=k\leq n$, and a computational basis $x\in \{0, 1\}^k$ 
    corresponding to the measurement operator 
    \leqalign{}{
        O(K, x) 
        = \df 1 {2^k} \prod_{j=1}^k I + (-1)^{x_j} Z_{K_j}, \quad 
        K\subset[2n], \quad x\in \{0, 1\}^{k}
    }
    which is $|x\ra \la x|$ when restricted 
    to the lines indexed by $K$. The expectation value of the measurement is 
    \leqalign{}{
        \tr\left[O(K, x)\rho\right] 
        = \df 1 {2^k} \sqrt{\det\left[I + \Sigma(\rho) \Sigma(K, x)\right]}.  
    }
    Here $\Sigma(K, x)$ is the antisymmetric covariance matrix 
    associated with the measurement defined by 
    \leqalign{}{
        \Sigma(K, x) = -i\sum_{j=1}^k (-1)^{x_j} s_{(2K_j+1)(2K_j+2)} \in \mfk{so}(2n, \C), 
        \quad s_{jk} = |j\ra \la k| - |k\ra \la j| 
    }
    \begin{proof}
        Fixing $K$ and $x$, 
        let $O=O(K, x)$. 
        Note that $I/2$ is the maximally mixed state, so 
        $O = 2^{n-k} \rho_O$ where $\rho_O = 2^{k-n} O\in \dgauss$. 
        The covariance matrix of $\rho(O)$ is 
        \malign{
            \Sigma(\rho_O)
            = -i \sum_{j=1}^k (-1)^{x_j}s_{(2K_j+1)(2K_j+2)}
             = \Sigma(K, x). 
        }
        Applying lemma~\ref{lem:dgOverlapFormula} concludes the proof 
        since 
        $ 
            \tr\left[O(K, x)\rho\right] = 2^{n-k} \tr(\rho_O \rho) 
            = 2^{n-k} \df{1}{2^n} 
            \sqrt{\det\left[I + \Sigma(\rho) \Sigma(K, x)\right]}$. 
    \end{proof}
\end{lemma}

\section{Unitary embedding}
\label{app:unitaryEmbedding}

This section begins by describing the first 
Gaussianity-preserving unitary embedding of $n$-qubit displaced 
Gaussian states into $(n+1)$-qubit even Gaussian states. 
Leveraging this tool, we generalize a previous work~\cite{lyu2024fermionic} 
on fermionic convolution and Gaussian testing to displaced Gaussian states. 
This yields useful operational protocols for testing efficiently 
simulable displaced Gaussian components as well as other characterizations 
of displaced Gaussian states. 

\begin{definition}[even embedding channel]
    \label{def:evenEmbeddingChannel}
    The even embedding channel $\mca E:\Cl_{2n}\to \Cl_{2n+2}$ 
    is defined by 
    \begin{eqnarray}
        \mca E(\rho) = V(\rho \ot |+\ra \la +|)V^\dag 
    \end{eqnarray}
    where $V\in \Cl_{2n+2}$ is the displaced Gaussian unitary 
    defined by 
    \begin{eqnarray}
        V = \exp \left(-i \df \pi 4 \gamma_{2n+2}\right). 
    \end{eqnarray}
\end{definition}
One can verify that $V$ effects the following 
transform of the Majorana generators exactly as $\phi$ in 
definition \ref{def:clifAlgEmbedding}: 
\leqalign{eq:evenConvUnitaryEffect}{
    V\gamma_j V^\dag = \phi(\gamma_j) 
    = \begin{cases}
        i \gamma_j \gamma_{2n+2} & j < 2n+2 \\ 
        \gamma_{2n+2} & j = 2n+2
    \end{cases}. 
}
To demonstrate that this is a desirable even embedding, 
we need the following lemma: 
\begin{lemma}[adjoining $|+\ra$ to an even Gaussian state]
    \label{lem:plusAdjGaussian}
    Given an even Gaussian state $\rho \in \gauss(n)$ with 
    $M=\Sigma(\rho)$, 
    the product state $\sigma = \rho \otimes |+\ra \la +|$ 
    is a displaced Gaussian state with extended covariance 
    \leqalign{}{
        \tilde \Sigma(\sigma) 
        = \tilde \Sigma(\rho \ot I/2) + i s_{(2n+1)(2n+3)}
        \in \mfk{so}(2n+3, \C), \quad 
        s_{jk} = |j\ra \la k| - |k\ra \la j|. 
    } 
    \begin{proof}
        Define $\tau = |+\ra \la +|\otimes \rho$; we obtain the 
        expansion for $\Xi_\tau$ as 
        \malign{
            \Xi_{\tau}(\theta) 
            &= \Xi_{|+\ra \la +|}(\theta) \otimes \Xi_\rho(\theta)
            = \exp(\theta_1) \otimes \exp \left(
                \sum_{j<k=1}^{2n}M_{jk}\theta_{j}\theta_{k}
            \right) \\
            &= \exp \left(\theta_1 + 
            \sum_{j<k=1}^{2n} M_{jk}\theta_{j+2}\theta_{k+2}\right). 
        }
        Note that the first equality relied on $\rho$ being even. 
        Recalling the fermionic swap unitary (equation~\ref{def:fswap}) 
        specialized to $S=S_{1\leftrightarrow 2n+1}$, 
        conjugation swaps the subspaces 
        $(1, 2)\leftrightarrow (2n+1, 2n+2)$ in the 
        extended covariance matrix representation, then 
        \malign{
            (S_{1\leftrightarrow 2n+1} \tau S_{1\leftrightarrow 2n+1}^\dag)
            (\theta) =  \exp \left(\theta_{2n+1} + 
            \sum_{j<k=1}^{2n} M_{jk}\theta_j\theta_k\right). 
        }
        We further have $\sigma = S_{1\leftrightarrow 2n+1} \tau S_{1\leftrightarrow 2n+1}^\dag$
        since fermionic swap $S_{1\leftrightarrow 2n+1}$ acts as a genuine 
        swap gate if any of the two lines has definite parity (i.e. is part 
        of an even state) and $\rho$ is even. Matching the Fourier 
        expression with the extended covariance matrix concludes the proof. 
    \end{proof}
\end{lemma}
\begin{theorem}[restatement of theorem~\ref{thm:unitaryEmbeddingMain}]
    \label{thm:dgStateEmbedding}
    Given a $n$-qubit state $\rho$, 
    $\mca E(\rho)\in \gauss(n+1)\iff \rho\in \dgauss(n)$, 
    in which case the covariance matrix of the embedding is 
    \leqalign{eq:evenEmbeddingCov}{ 
        \Sigma[\mca E(\rho)] = \begin{bmatrix}
            \Sigma(\rho) & -ir & i\mu(\rho) \\ 
            ir^T & 0 & ic \\ 
            -i\mu(\rho)^T & -ic &  0
        \end{bmatrix}\in \mfk{so}(2n+2), \quad 
        R = \begin{bmatrix}
            R_0 & s \\ r^T & c 
        \end{bmatrix} \in SO(2n+1). 
    }
    Here $r\in \R^{2n}$ and $c\in \R$ are the entries of $R$ such that 
    $R\, \tilde \Sigma(\rho) R^T$ is block-diagonal and trivial on the last subspace. 
    Moreover, $\rho$ is displaced Gaussian iff $\mca E(\rho)$ is Gaussian. 
    \begin{proof}
        Let $\Sigma = \Sigma(\rho), \mu=\mu(\rho)$. 
        We first prove that that $\rho \ot |+\ra \la +|$ is a 
        displaced Gaussian state: 
        by theorem~\ref{thm:dgStateCharacterization}, write 
        $\rho = U_G^\dag \rho_D U_G$ for $U_G\in DG(n)$ 
        then $U_G$ affects $R$ that block-diagonalizes 
        $\tilde \Sigma(\rho)$. Then $\rho = U_G^\dag \rho_D U$ implies 
        \[ 
            \rho \ot |+\ra \la +| = 
            (U_G^\dag \ot I) (\rho_D \ot |+\ra \la +|) (U_G\ot I)
        \] 
        is Gaussian since $\rho_D \ot |+\ra \la +|$ is Gaussian by 
        lemma~\ref{lem:plusAdjGaussian}. Decompose $R$ in terms 
        of $R_0, r, c$ as in equation~\ref{eq:evenEmbeddingCov}, 
        the rotation $\bar R\in SO(2n+3)$ corresponding to $U_G\ot I$, as 
        well as the extended covariance matrix $\tilde \Sigma(\rho \ot I)$, are 
        \malign{
            \bar R = \begin{bmatrix}
                R_0 & 0_{2n\times 2} & s \\ 
                0_{2\times 2n} & I_{2\times 2} & 0_{1\times 2} \\ 
                r^T & 0_{2\times 1} & c 
            \end{bmatrix}, \quad 
            \tilde \Sigma(\rho \ot I/2) = \begin{bmatrix}
                \Sigma & 0_{2n\times 2} & i\mu \\ 
                0_{2\times 2n} & 0_{2\times 2} & 0 \\ 
                -i\mu^T & 0 & 0 
            \end{bmatrix} = \bar R^T\, \tilde \Sigma(\rho_D \ot I/2) \bar R. 
        }
        Recalling $\tilde \Sigma(\rho_D \ot |+\ra \la +|) = \tilde \Sigma(\rho_D \ot I) 
        + is_{(2n+1)(2n+3)}$ from 
        lemma~\ref{lem:plusAdjGaussian}, we obtain 
        \malign{
            \tilde \Sigma(\rho \ot |+\ra \la +|)
            &= \bar R^T \left[
                \tilde \Sigma(\rho_D \ot |+\ra \la +)
            \right] \bar R = \tilde \Sigma(\rho \ot I/2) + i \bar R^T \left(
                |2n+1\ra \la 2n+3| - |2n+3\ra \la 2n+1|
            \right) \bar R  \\ 
            &= \tilde \Sigma(\rho \ot I/2) + i \left(
                |2n+1\ra \begin{bmatrix}
                    r \\ 0_{1\times 2} \\ c 
                \end{bmatrix} - \begin{bmatrix}
                    r^T & 0_{2\times 1} & c 
                \end{bmatrix}\la 2n+1|
            \right) \\ 
            &= \begin{bmatrix}
                \Sigma & -ir & 0_{2n\times 1} & i\mu \\ 
                ir^T & 0 & 0 & ic \\ 
                0_{1\times 2n} & 0 & 0 & 0 \\ 
                -i\mu^T & -ic & 0 & 0 
            \end{bmatrix} \implies 
            \tilde \Sigma[V(\rho \ot |+\ra \la +|)V^\dag] 
            = \begin{bmatrix}
                \Sigma & -ir & i\mu & 0_{2n\times 1}\\ 
                ir^T & 0 & ic & 0 \\ 
                -i\mu^T & -ic & 0 & 0  \\ 
                0_{1\times 2n} & 0 & 0 & 0
            \end{bmatrix}
        }
        The last implication holds because 
        conjugation by $V$ swaps the $(2n+2)$ and 
        $(2n+3)$-th subspaces in $\tilde \Sigma(\rho \ot |+\ra \la +|)$ 
        using equation~\ref{eq:evenConvUnitaryEffect}. 
        We obtain the desired covariance matrix relation 
        by noting that $\Sigma$ 
        is the restriction to all but the last subspace. 
        This establishes the forward direction in the equivalence; 
        for the converse, suppose for contraposition that 
        $\mca E(\rho) = V(\rho \otimes |+\ra \la +|) V^\dag$ 
        is non-Gaussian; since $V\in DG(n+1)$, this 
        implies that $\rho \otimes |+\ra \la +|$ is non-Gaussian. 
        But $\rho_G\otimes |+\ra \la +|$ is displaced Gaussian if $\rho_G$ is 
        Gaussian by the reasoning above, so $\rho$ cannot be displaced Gaussian. 
    \end{proof}
\end{theorem} 

The even Gaussian state embedding has a corresponding 
compatible embedding for Gaussian unitaries. 
\begin{lemma}[restatement of lemma \ref{lem:unitaryEmbeddingMain}]
    Let $U\in \Cl_{2n}$, define $\tilde U = V (U\otimes I)V^\dag \in \Cl_{2n+2}$, 
    with $V$ being the same unitary in the even embedding 
    channel~\ref{def:evenEmbeddingChannel}, then for any $\tilde U$ in the 
    image of the transformation, 
    \begin{eqnarray}
        U\in DG(n)\iff \tilde U \in G(n+1). 
    \end{eqnarray}
    The mean $i\gamma^Td$ is transformed into 
    covariance $-i (\gamma^T d)\gamma_{2n+2}$, and 
    the Gaussian expressions are 
    \leqalign{eq:dgUnitaryEmbedding}{
        U = \exp \left(\df 1 2 \gamma^T h \gamma + i\gamma^T d\right) 
        \iff \tilde U = \exp\left(\df 1 2 \gamma^T \tilde h \gamma\right), \quad 
        \tilde h = \begin{bmatrix}
            h & 0_{2n\times 1} & -d \\ 0_{1\times 2n} & 0 & 0 \\ 
            d^T & 0 & 0  
        \end{bmatrix}. 
    }
    \begin{proof}
        Note $\phi(\gamma_j) = V\gamma_j V^\dag$. Since conjugation by $V$ 
        is unitary, it is extended multiplicatively just as $\phi$. 
        Let $\gamma = (\gamma_1, \dots, \gamma_{2n})$ and 
        $\tilde \gamma = (\gamma_1, \dots, \gamma_{2n+2})$, then 
        \malign{
            \tilde U &= \phi(e^{\log U}) 
            = \exp \left[
                \phi\left(
                    \df 1 2 \gamma^T h \gamma + i \gamma^T d
            \right)\right] 
            = \exp \left[\df 1 2 \gamma^T h \gamma - (\gamma^T d)\gamma_{2n+1}\right]
            = \exp \left(
                \df 1 2 \tilde \gamma^T \tilde h \tilde \gamma
            \right). 
        }
    \end{proof}
\end{lemma}

We now proceed to decompose $V\in DG(n+1)$ into products of elementary gates. 
Recalling $S=\mrm{diag}(1, i)$, let 
$A=S^\dag H$ and $\mrm{CX}_{a\to b}$ denote the controlled-not 
with $a$ as control and $b$ as target. This yields the conjugation relations 
\begin{eqnarray}
    A Y A^\dag = -Z, \quad (\mrm{CX}_{a\to b}) \, Z_b \, 
    (\mrm{CX}_{a\to b}) = Z_aZ_b. 
\end{eqnarray}
Fixing $n$, let $B = \prod_{j=1}^n \mrm{CX}_{j\to n+1}$
and letting $\cong$ denote equivalence up to a global phase, we obtain: 
\malign{
    AB S_{n+1} BA^\dag 
    &\cong AB \exp \left(i \df \pi 4 Z_{n+1}\right) BA^\dag 
    = A \exp \left(i \df \pi 4 Z_1\dots Z_{n+1}\right) A^\dag \\ 
    &= \exp \left(-i \df \pi 4 Z_1\dots Z_n Y_{n+1}\right) 
    = \exp \left(-i \df \pi 4 \gamma_{2n+2}\right) = V\in DG(n). 
}
\begin{figure}[t]
    \center{\includegraphics[width=0.8\linewidth]  {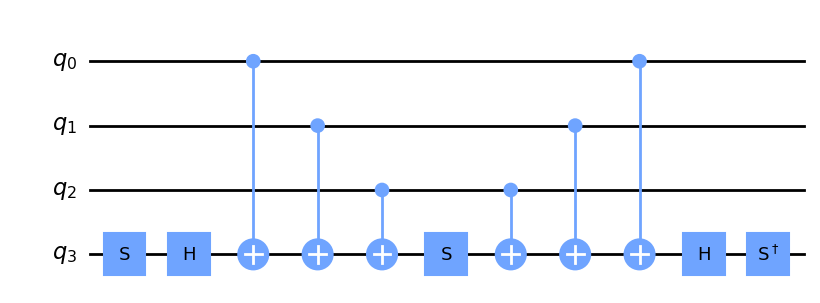}}     
    \caption{Elementary decomposition of the even embedding unitary $V$ 
    (equation~\ref{def:embeddingUnitary}) used to embed $\dgauss(3)$. }
    \label{fig:embedUnitary}
\end{figure}
The $n=3$ case is illustrated in Fig.~\ref{fig:embedUnitary} using the Qiskit 
library~\cite{wille2019ibm}. 




\end{document}